\documentclass[conference,letterpaper]{IEEEtran}

\usepackage[utf8]{inputenc} 
\usepackage[T1]{fontenc}
\usepackage{url}
\usepackage{ifthen}
\usepackage{cite}
\usepackage{amsmath,amssymb, amsthm, amsfonts, amsxtra, bbm, mathrsfs, stmaryrd}
\usepackage[colorlinks,linkcolor=red,citecolor=blue,urlcolor=blue]{hyperref}
\newtheorem{theorem}{Theorem}
\newtheorem{lemma}{Lemma}
\newtheorem{corollary}{Corollary}
\newtheorem{example}{Example}
\newtheorem{remark}{Remark}

\begin{document}
\title{Constrained Functional Value under General Convexity Conditions with Applications to Distributed Simulation} 


\author{\IEEEauthorblockN{Yanjun Han}
\IEEEauthorblockA{Department of Electrical Engineering, Stanford University\\
Email: yjhan@stanford.edu}
}



\maketitle

\begin{abstract}
We show a general phenomenon of the constrained functional value for densities satisfying general convexity conditions, which generalizes the observation in \cite{bobkov2011entropy} that the entropy per coordinate in a log-concave random vector in any dimension with given density at the mode has a range of just 1. Specifically, for general functions $\phi$ and $\psi$, we derive upper and lower bounds of density functionals taking the form $I_\phi(f) = \int_{\mathbb{R}^n} \phi(f(x))dx$ assuming the convexity of $\psi^{-1}(f(x))$ for the density, and establish the tightness of these bounds under mild conditions satisfied by most examples. We apply this result to the distributed simulation of continuous random variables, and establish an upper bound of the exact common information for $\beta$-concave joint densities, which is a generalization of the log-concave densities in \cite{li2017distributed}. 
\end{abstract}


\section{Introduction}
This paper is motivated by the following observation in \cite{bobkov2011entropy}: for any log-concave density $f$ on $\mathbb{R}^n$, its differential entropy $h(f) = \int_{\mathbb{R}^n} -f(x)\log f(x)dx$ satisfies
\begin{align}\label{eq.bobkov}
\log\left(\frac{1}{f_{\max}}\right)\le h(f) \le \log\left(\frac{1}{f_{\max}}\right) + n,
\end{align}
where $f_{\max} = \sup_{x\in \mathbb{R}^n} f(x)$ is the sup-norm of the density $f$ (throughout $\log$ is assumed to be the natural logarithm). This inequality shows that the differential entropy of a random vector is constrained in a limited range, provided that the density is log-concave. The main contribution of this paper is to generalize the inequality \eqref{eq.bobkov} to general density functionals and convexity conditions on the density, and therefore show a general phenomenon of the constrained functional value under convexity conditions. Specifically, 
\begin{itemize}
	\item We derive tight upper and lower bounds for $\phi$-functionals of $\psi$-concave densities involving the mode of the density, where $\phi$ and $\psi$ are general functions. We also establish the tightness of these bounds under mild conditions, and develop stronger or new inequalities for specific choices of $(\phi, \psi)$. 
	\item We generalize the result in \cite{li2017distributed} on distributed simulation of continuous random variables, and show that with generalizations of \eqref{eq.bobkov} to other convexity conditions, similar upper bounds on the exact common information are available for convexity conditions weaker than log-concavity. 
\end{itemize}

Convexity properties of the probability density have been deeply studied in probability, statistics and geometry, where there exist functional inequalities (e.g., Poincar\'{e} and logarithmic Sobolev inequalities) for either the log-concave \cite{prekopa1973logarithmic,bakry1985diffusions,bobkov1999isoperimetric} or $\beta$-concave measures \cite{brascamp2002extensions}. In information theory, inequalities in the same spirits as \eqref{eq.bobkov} include the reverse entropy power inequality \cite{bobkov2012reverse,xu2016reverse} and the reverse R\'{e}nyi entropy power inequality \cite{li2018renyi}, for both the log-concave and $\beta$-concave measures. However, most of these results are established in an ad-hoc fashion involving other deep inequalities (such as the Borell's theorem \cite{borell1973complements}), and it will be helpful to develop elementary and general tools for such inequalities. 

This paper is organized as follows. In Section \ref{sec.thm}, we present the general theorem and develop tight inequalities for several density functionals and convexity conditions. In Section \ref{sec.distributed}, we review the problem of simulating continuous random variables in a distributed manner, and establish an upper bound of the exact common information for $\beta$-concave joint densities. 

\section{Constrained Functional Value under Convexity Conditions}\label{sec.thm}
In this section, we formally state the upper and lower bounds for density functionals under general convexity conditions on the density, and establish the tightness of these bounds. Then we apply them to various examples of density functionals and convexity conditions, recover or strengthen several previously known inequalities, and prove new inequalities. 

\subsection{General Theorem}
Let $f$ be a probability density function on $\mathbb{R}^n$, and one is interested in the following $\phi$-functional of $f$: 
\begin{align}\label{eq.phi_functional}
I_\phi(f) \triangleq \int_{\mathbb{R}^n} \phi(f(x))dx,
\end{align}
where $\phi: \mathbb{R}_+\to \mathbb{R}$ is a given absolutely continuous function with $\phi(0) = 0$. For example, the functional $I_\phi(f)$ with $\phi(t) = -t\log t$ is the differential entropy $h(f)$ of $f$. Our target is to find tight upper and lower bounds of the $\phi$-functional $I_\phi(f)$, which may depend on the following sup-norm of the density to handle possible scalings of $f$:
\begin{align}\label{eq.sup-norm}
f_{\max} \triangleq \sup_{x\in \mathbb{R}^n} f(x). 
\end{align} 

To restrict the possibly large class of the density, we impose convexity conditions on $f$ and assume that $f$ is \emph{$\psi$-convex}, i.e., 
\begin{align}\label{eq.psi_convex}
f(x) = \psi(g(x)),
\end{align}
where $\psi: (a,\infty) \to \mathbb{R}_+$ is a given continuous and strictly decreasing function with $\lim_{t\to\infty} \psi(t) = 0$ (it is possible that $a=-\infty$), and $g: \mathbb{R}^d \to (a,\infty]$ is some convex function. Note that \eqref{eq.psi_convex} is a general convexity condition on the density $f$, where the choices $\psi(t) = e^{-t}$ and $\psi(t) = t^{-\beta}$ with $\beta>0$ lead to \emph{log-concave} and \emph{$\beta$-concave} densities, respectively. 
\begin{remark}
The notion of $\beta$-concavity in this paper is slightly different from the $\kappa$-concavity used in some literature (where $x\mapsto f(x)^\kappa$ is convex), with the correspondence $\kappa = -\beta^{-1}$. 
\end{remark}

Our first theorem derives upper and lower bounds of general $\phi$-functionals for any $\psi$-convex densities. 
\begin{theorem}\label{thm.upper}
Under the above setup \eqref{eq.phi_functional}-\eqref{eq.psi_convex}, let $b \triangleq \psi^{-1}(f_{\max})$ and $F_k, G_k: [b,\infty)\to \mathbb{R}$ be real-valued functions vanishing at the infinity such that  (assuming the existence of $(F_k, G_k)$)
\begin{align*}
(-1)^k \frac{d^k}{dx^k}F_k(x) &= \phi(\psi(x)), \\
(-1)^k \frac{d^k}{dx^k}G_k(x) &= \psi(x), 
\end{align*}
for $k=0,1,\cdots,n$. If there exists real number $A$ such that: 
\begin{enumerate}
	\item[(i)] $F_n(b) - A\cdot G_n(b) \le 0$; 
	\item[(ii)] $F_0(b) - A\cdot G_0(b) \le 0$; 
	\item[(iii)] The function $x\mapsto F_0(x) - A\cdot G_0(x)$ has at most one zero on $[b,\infty)$; 
\end{enumerate}
then $I_\phi(f)\le A$. Similarly, $I_\phi(f)\ge A$ if both $\le$ in conditions (i) and (ii) are replaced by $\ge$. 
\end{theorem}

Theorem \ref{thm.upper} provides upper and lower bounds for the $\phi$-functional $I_\phi(f)$ depending only on $\phi, \psi$, and on the density $f$ only through its sup-norm $f_{\max}$. Although the choice of $A$ in Theorem \ref{thm.upper} seems complicated (involving the antiderivatives), the next result shows that this choice is essentially tight. 

\begin{theorem}\label{thm.lower}
Under the setting of Theorem \ref{thm.upper}, if for some real number $A$ either the condition (i) or (ii) is violated, then there exists a density $f$ on $\mathbb{R}^n$ with sup-norm $f_{\max}$ and $I_\phi(f)>A$. 
\end{theorem}
\begin{corollary}\label{cor}
If $A = \max\{F_n(b)/G_n(b), F_0(b)/G_0(b)\}$ satisfies the condition (iii) of Theorem \ref{thm.upper}, then
\begin{align*}
\sup\left\{I_\phi(f): f \text{ is } \psi\text{-convex} \right\} = A. 
\end{align*}
Similarly, if $A' = \min\{F_n(b)/G_n(b), F_0(b)/G_0(b)\}$ satisfies the condition (iii) of Theorem \ref{thm.upper}, then
\begin{align*}
\inf\left\{I_\phi(f): f \text{ is } \psi\text{-convex} \right\} = A'. 
\end{align*}
\end{corollary}

Theorem \ref{thm.lower} shows that the conditions (i) and (ii) in Theorem \ref{thm.upper} are necessary, and thus the $n$ times repeated antiderivatives of $\phi\circ \psi$ and $\psi$ play fundamental roles on the constrained value of the density functional. Condition (iii) is an additional technical restriction on $\phi\circ \psi$ and $\psi$, which typically holds and is easily verifiable for most examples. In the next subsection, we will show through various examples that Theorem \ref{thm.upper} gives tight inequalities on both sides and is easy to evaluate when the functions $F_k, G_k$ admit closed-form expressions.

\subsection{Examples}
Throughout the following examples, we will either assume that $\psi(x) = e^{-x}$ on $\mathbb{R}$ or $\psi(x) = x^{-\beta}$ on $(0,\infty)$, representing log-concave and $\beta$-concave densities, respectively. The $k$ times anti-derivatives of these functions are
$
G_k(x) = e^{-x} 
$
and 
$$
G_k(x) = \frac{x^{k-\beta}}{(\beta-1)(\beta-2)\cdots (\beta-k)},
$$
respectively, provided that $\beta>k$. 

\begin{example}[Differential entropy]\label{example.entropy}
Let $\phi(x) = -x\log x$ so that $I_\phi(f) = h(f)$ is the differential entropy of $f$. Then for $\psi(x) = e^{-x}$, we have $b = \log(1/f_{\max}), \phi(\psi(x)) = xe^{-x}$ and 
$$
F_k(x) = (x+k)e^{-x}, \qquad k=0,1,\cdots,n. 
$$
Since $F_0(x) - A\cdot G_0(x) = (x-A)e^{-x}$ has at most one zero on $\mathbb{R}$ for any $A$, the condition (iii) of Theorem \ref{thm.upper} is satisfied. Hence, by Corollary \ref{cor}, the following inequality holds: 
\begin{align}\label{eq.entropy_log}
\log\left(\frac{1}{f_{\max}}\right)\le h(f) \le \log\left(\frac{1}{f_{\max}}\right) + n,
\end{align}
and both inequalities are tight. This inequality recovers \eqref{eq.bobkov}.

For $\beta$-concave densities with $\beta>n$, we have $b = f_{\max}^{-1/\beta}$, $\phi(\psi(x)) = \beta x^{-\beta}\log x$, and 
$$
F_k(x) = \frac{\beta x^{k-\beta}(\log x + \sum_{i=1}^k (\beta-i)^{-1} )}{(\beta-1)(\beta-2)\cdots (\beta-k)}, \quad k=0,1,\cdots,n. 
$$
Again, the function $F_0(x) - A\cdot G_0(x) = (\beta\log x- A)x^{-\beta}$ has at most one zero for any $A$. Hence, Corollary \ref{cor} gives
\begin{align}\label{eq.entropy_beta}
\log\left(\frac{1}{f_{\max}}\right)\le h(f) \le \log\left(\frac{1}{f_{\max}}\right) + \sum_{i=1}^n \frac{\beta}{\beta - i},
\end{align}
and both inequalities are tight. Note that \eqref{eq.entropy_beta} reduces to \eqref{eq.entropy_log} by taking $\beta\to\infty$. This inequality fills the gap of \cite[Theorem I.3]{bobkov2011entropy} from $\beta\ge n+1$ to any $\beta>n$, and recovers \cite[Corollary 7.1]{fradelizi2020concentration}.
\end{example}

\begin{example}[R\'{e}nyi entropy]
Let $\phi(x) = x^\alpha$ with $\alpha>0, \alpha\neq 1$, then $\frac{1}{1-\alpha}\log I_\phi(f)$ is the $\alpha$-R\'{e}nyi entropy $h_\alpha(f)$ of $f$. For log-concave densities, we have $b = \log(1/f_{\max})$, $\phi(\psi(x)) = e^{-\alpha x}$, and
$$
F_k(x) = \alpha^{-k}e^{-\alpha x}, \quad k=0,1,\cdots,n. 
$$
Clearly, $F_0(x) - A\cdot G_0(x) = (e^{(1-\alpha)x} - A)e^{-x}$ has at most one zero for any $A$. Hence, Corollary \ref{cor} gives
\begin{align}\label{eq.renyi_log}
\log\left(\frac{1}{f_{\max}}\right)\le h_\alpha(f) \le \log\left(\frac{1}{f_{\max}}\right) + \frac{n\log\alpha}{\alpha-1},
\end{align}
and both inequalities are tight. This is a generalization of \cite[Theorem IV.1]{bobkov2011entropy} from $\alpha\in (1,\infty)$ to the entire nonnegative axis $\alpha\in (0,\infty)\backslash \{1\}$, and recovers \cite[Corollary 7.1]{fradelizi2020concentration}.

For $\beta$-concave densities with $\min\{\alpha,1\}\cdot \beta>n$, we have $b = f_{\max}^{-1/\beta}$, $\phi(\psi(x))=x^{-\alpha\beta}$, and 
$$
F_k(x) = \frac{x^{k-\alpha\beta}}{(\alpha\beta-1)(\alpha\beta-2)\cdots(\alpha\beta-n)}, \quad k=0,1,\cdots,n. 
$$
Again, $F_0(x) - A\cdot G_0(x) = (x^{(1-\alpha)\beta}-A)x^{-\beta}$ has at most one zero for any $A$. Hence, Corollary \ref{cor} gives
\begin{align}\label{eq.renyi_beta}
\log\left(\frac{1}{f_{\max}}\right)\le h_\alpha(f) \le \log\left(\frac{1}{f_{\max}}\right) + \frac{1}{\alpha-1}\sum_{i=1}^n \log \frac{\alpha\beta-i}{\beta-i}, 
\end{align}
and both inequalities are tight. Note that \eqref{eq.renyi_beta} reduces to \eqref{eq.renyi_log} by taking $\beta\to\infty$, and the R\'{e}nyi entropy inequalities reduce to the differential entropy ones by taking $\alpha\to 1$. The inequality \eqref{eq.renyi_beta} is a generalization of \cite[Theorem VIII.1]{bobkov2011entropy} from $\alpha\in (1,\infty)$ to $\alpha\in (0,\infty)\backslash \{1\}$, and recovers \cite[Corollary 7.1]{fradelizi2020concentration}.
\end{example}

\begin{example}[Truncated density]\label{example.truncated}
Another interesting functional is the truncation function $\phi_t(x) = \min\{x,t\}$, where $0<t< f_{\max}$ is a given threshold. Consequently, $I_{\phi_t}(f) = \int \min\{f(x),t\}dx$ is the remaining total probability if the density $f$ is truncated at $t$. The target is to understand how fast the total probability $I_{\phi_t}(f)$ decays as $t$ approaches zero. 

For log-concave densities, we have $b=\log(1/f_{\max})$, $\phi_t(\psi(x)) = \min\{e^{-x}, t\}$. Simple algebra shows that the function $F_0(x) - A\cdot G_0(x) = \min\{e^{-x},t\}-Ae^{-x}$ has at most one zero as long as $A<1$. Moreover, by Cauchy formula for repeated integration, we have
$$
F_n(b) = \int_b^\infty \frac{(x-b)^{n-1}}{(n-1)!}\min\{e^{-x}, t\}dx. 
$$
Some elementary but tedious algebra lead to
$$
\frac{F_n(b)}{G_n(b)} = \frac{t}{f_{\max}}\sum_{k=0}^n \frac{1}{k!}\left(\log \frac{f_{\max}}{t}\right)^k. 
$$
Hence, by Corollary \ref{cor}, we have the tight inequality
\begin{align}\label{eq.truncation_log}
\frac{t}{f_{\max}} \le I_{\phi_t}(f) \le \frac{t}{f_{\max}}\sum_{k=0}^n \frac{1}{k!}\left(\log \frac{f_{\max}}{t}\right)^k, 
\end{align}
strengthening the results appeared in \cite[Lemma 4]{li2017distributed}. Moreover, note that the RHS of \eqref{eq.truncation_log} is simply $\mathbb{P}(\mathsf{Poi}(\log(f_{\max}/t)) \le n)<1$, i.e., the Poisson CDF at $n$ with rate parameter $\log(f_{\max}/t)$. 

For $\beta$-concave densities with integer $\beta\ge n+1$, we have $b=f_{\max}^{-1/\beta}$, and $\phi_t(\psi(x)) = \min\{x^{-\beta},t\}$. Similarly, one can show that $F_0(x) - A\cdot G_0(x)$ has at most one zero for any $A<1$, and Cauchy formula for repeated integration leads to
\begin{align*}
\frac{F_n(b)}{G_n(b)} &= \left(\frac{t}{f_{\max}}\right)^{1-\frac{n}{\beta}}\sum_{k=0}^n \frac{\Gamma(\beta-n+k)}{\Gamma(\beta-n)k!}\left(1 - \left(\frac{t}{f_{\max}}\right)^{\frac{1}{\beta}}\right)^k \\
&= \mathbb{P}\left(\mathsf{NB}\left(\beta-n, 1-\left(\frac{t}{f_{\max}}\right)^{\frac{1}{\beta}}\right) \le n\right) \\
&= \mathbb{P}\left(\mathsf{B}\left(\beta, 1-\left(\frac{t}{f_{\max}}\right)^{\frac{1}{\beta}}\right) \le n\right) < 1,
\end{align*}
where $\Gamma(z)=\int_0^\infty t^{z-1}e^{-t}dt$ denotes the Gamma function, $\mathsf{NB}(r,p)$ denotes the negative binomial distribution with success probability $p$ and number of failures $r$ until the experiment is stopped, and $\mathsf{B}(r,p)$ denotes the Binomial distribution with total number of trials $r$ and success probability $p$. Note that the last equality is simply the relationship between direct and inverse samplings \cite{morris1963note}. Consequently, Corollary \ref{cor} gives the following tight inequality: 
\begin{align}\label{eq.truncation_beta}
\frac{t}{f_{\max}} \le I_{\phi_t}(f) \le \mathbb{P}\left(\mathsf{B}\left(\beta, 1-\left(\frac{t}{f_{\max}}\right)^{\frac{1}{\beta}}\right) \le n\right).  
\end{align}
Inequalities \eqref{eq.truncation_log} and \eqref{eq.truncation_beta} are related via the following Poisson CLT result
$
\mathsf{B}(r, 1 - e^{-\lambda/r}) \overset{d}{\to} \mathsf{Poi}(\lambda)
$ as $r\to\infty$. 
Moreover, by Corollary \ref{cor}, \eqref{eq.truncation_beta} gives the following variational representation of the Binomial CDF. 
\end{example}
\begin{lemma}\label{lemma.bino}
For integer $n\ge 1, k\ge 0$ and real number $c>0$, the following identity holds: let $\mathcal{F}_{n,k}$ be the set of all $n$-concave densities on $\mathbb{R}^k$ with $\sup_{x\in \mathbb{R}^k} f(x) = 1$, then 
$$
\mathbb{P}(\mathsf{B}(n, 1 - e^{-c/n}) \le k) = \sup_{f\in \mathcal{F}_{n,k}} \int_{\mathbb{R}^k} \min\{f(x),e^{-c}\}dx. 
$$
In particular, since $\mathcal{F}_{n+1,k}\subseteq \mathcal{F}_{n,k}$, the map $n\mapsto \mathbb{P}(\mathsf{B}(n, 1 - e^{-c/n}) \le k) $ is non-increasing.
\end{lemma}

\subsection{Proof of Theorem \ref{thm.upper}}
For $k=0,1,\cdots,n$ and given $A$ satisfying the conditions of Theorem \ref{thm.upper}, let $H_k(x) = F_k(x) - A\cdot G_k(x)$. Since all functions $H_k(x)$ vanish at the infinity and $H_{k+1}'(x) = -H_k(x)$, the Cauchy formula for repeated integration gives
\begin{align}\label{eq.cauchy}
H_k(x) = \int_{x}^\infty \frac{(y-x)^{k-1}}{(k-1)!}H_0(y)dy, \quad k=1,\cdots,n. 
\end{align}
The next lemma summarizes some key properties of $H_k$. 
\begin{lemma}\label{lemma.Hk}
Under the conditions of Theorem \ref{thm.upper}, for each $k=0,1,\cdots,n$ the following holds: $H_k(b)\le 0$ and $H_k(x)$ has at most one zero on $[b,\infty)$. 
\end{lemma}
\begin{proof}
We first show that $H_k(b)\le 0$. By condition (ii) and (iii), $H_0(x)$ has at most one zero on $[b,\infty)$ and $H_0(b)\le 0$. If $H_0(x)$ has no zeros on $[b,\infty)$, the continuity of $H_0(\cdot)$ implies that $H_0(x)\le 0$ for all $x\in [b,\infty)$, then $H_k(b)\le 0$ is a direct consequence of \eqref{eq.cauchy}. If $H_0(x)$ has exactly one zero $c\in [b,\infty)$, then by continuity again we must have $H_0(x)\le 0$ if $x\le c$ and $H_0(x)\ge 0$ if $x\ge c$. Hence, for all $x\in [b,\infty)$ we have
$$
\left[(x-b)^{k-n} - (c-b)^{k-n} \right]H_0(x) \le 0. 
$$
Consequently, \eqref{eq.cauchy} with $x=b$ gives
\begin{align*}
H_k(b) &= \int_b^{\infty} \frac{(x-b)^{k-1}}{(k-1)!}H_0(x)dx \\
&\le \int_b^{\infty} \frac{(x-b)^{n-1}}{(k-1)!}(c-b)^{k-n}H_0(x)dx \\
&= \frac{(n-1)!(c-b)^{k-n}}{(k-1)!}\cdot \int_b^\infty \frac{(x-b)^{n-1}}{(n-1)!}H_0(x)dx \\
&= \frac{(n-1)!(c-b)^{k-n}}{(k-1)!}\cdot H_n(b) \le 0, 
\end{align*}
where the last step is due to condition (i). Hence in both cases we have $H_k(b)\le 0$. 

The second claim is proved via induction on $k$. The base case $k=0$ is simply the condition (iii). Now assume that $H_k(x)$ has at most one zero on $[b,\infty)$ for some $k\le n-1$ and we consider $H_{k+1}$. Since $H_k(b)\le 0$, the continuity of $H_k$ and the induction hypothesis imply that either $H_k(x)< 0$ for all $x\ge b$, or $H_k(x) < 0$ if $x\in [b,c)$ and $H_k(x)> 0$ if $x\in (c,\infty)$ for some real number $c$. Since $H_{k+1}'(x) = - H_k(x)$ and $H_{k+1}(b)\le 0, \lim_{x\to\infty} H_{k+1}(x)=0$, in the first scenario we must have $H_{k+1}(x)<0$ for all $x>b$, i.e., no zeros on $[b,\infty)$. In the second scenario, $H_{k+1}$ is first strictly increasing and then strictly decreasing on $[b,\infty)$, given the signs of $H_{k+1}$ on the end points we conclude that $H_{k+1}(x)$ has exactly one zero on $[b,\infty)$. This concludes the induction step and the second claim is also proved. 
\end{proof}

Now we go back to the proof of Theorem \ref{thm.upper}. For any density $f$ defined on $\mathbb{R}^b$ and $t>0$, define
$
f^\star(t) = \text{Vol}_n(\{x: f(x)\ge t\})
$
to be the $n$-dimensional volume of the $t$-superlevel set of $f$. Since $\phi$ is absolutely continuous with $\phi(0)=0$, the layer cake decomposition gives
\begin{align}\label{eq.layercake}
I_\phi(f) = \int_{0}^\infty \phi'(t) f^\star(t)dt. 
\end{align}

Since $\psi$ is strictly decreasing and $f(x)=\psi(g(x))$, we have
$
\{x: f(x)\ge \psi(u) \} = \{x: g(x) \le u\}
$
for any $u\ge b$. Recall that $g$ is convex, the Brunn--Minkowski inequality in $n$ dimensions implies that the mapping
$$
u\in [b,\infty) \mapsto h(u)\triangleq f^\star(\psi(u))^{1/n}
$$
is non-negative, non-decreasing and concave. Now a change of variable $t=\psi(u)$ in both \eqref{eq.layercake} and $\int_{\mathbb{R}^n} f(x)dx = 1$ gives
\begin{align}
I_\phi(f) &= \int_b^\infty h(u)^n \phi'(\psi(u))(-\psi'(u))du, \label{eq.I_phi} \\
1 &= \int_b^\infty h(u)^n (-\psi'(u))du. \label{eq.normalization}
\end{align}
Based on \eqref{eq.I_phi} and \eqref{eq.normalization}, the target inequality $I_\phi(f)\le A$ is equivalent to
\begin{align}\label{eq.target}
\int_b^\infty h(u)^n \frac{d}{du}\left[\phi(\psi(u)) - A\psi(u)\right] du \ge 0. 
\end{align}

By integration by parts, for $k=0,1,\cdots,n-1$ we have
\begin{align}
& \int_b^\infty h(u)^{n-k}H_k'(u)dx = -H_k(b)h(b)^{n-k} \nonumber\\
& \qquad - (n-k)\int_b^\infty h(u)^{n-k-1}h'(u) H_k(u)du \nonumber\\
& \qquad \ge -(n-k)\int_b^\infty h(u)^{n-k-1}h'(u) H_k(u)du, \label{eq.integration_parts}
\end{align}
where the last inequality is due to $h(b)\ge0$ and $H_k(b)\le 0$ by Lemma \ref{lemma.Hk}, and we note that by the non-decreasing property and concavity, the function $h$ is absolutely continuous on $[b,\infty)$ and is therefore differential almost everywhere. By Lemma \ref{lemma.Hk}, $H_k(x)$ has at most one zero on $[b,\infty)$. If $H_k(x)$ has no zeros, then by continuity $H_k(u)<0$ for all $u\in [b,\infty)$, and \eqref{eq.integration_parts} is non-negative since $h'(u)\ge 0$. If $H_k(x)$ has exactly one zero $x=c\in [b,\infty)$, then applying the same arguments of Lemma \ref{lemma.Hk} with the concavity of $h$ gives
$
(h'(u) - h'(c))H_k(u) \le 0 
$ for all $u\in [b,\infty)$. Hence, 
\begin{align}
&\int_b^\infty h(u)^{n-k-1}h'(u) H_k(u)du \nonumber \\
&\qquad \le \int_b^\infty h(u)^{n-k-1}h'(c) H_k(u)du \nonumber \\
&\qquad = - h'(c) \int_b^\infty h(u)^{n-k-1} H_{k+1}'(u)du. \label{eq.rolle}
\end{align}
Combining \eqref{eq.integration_parts} and \eqref{eq.rolle}, we conclude that in both cases, a sufficient condition for the inequality $\int_b^\infty h(u)^{n-k}H_k'(u)du\ge 0$ is $\int_b^\infty h(u)^{n-k-1}H_{k+1}'(u)du\ge 0$. Since the target inequality \eqref{eq.target} corresponds to $k=0$, it then suffices to prove the non-negativity for $k=n$, which reduces to
$
\int_b^\infty H_n'(u)du = - H_n(b) \ge 0, 
$
exactly the condition (i). The proof is complete. 

\subsection{Proof of Theorem \ref{thm.lower}}
If the condition (i) is violated, consider the density
\begin{align*}
f(x) = \psi(b+\lambda(x_1+x_2+\cdots+x_n)), \quad x_1,\cdots,x_n\ge 0
\end{align*}
with $b = \psi^{-1}(f_{\max})$ and $\lambda = G_n(b)^{1/n}$. It is clear that $f(x)$ is $\psi$-convex, and the function $f$ is a valid density by the choice of $\lambda$. Moreover, the decreasing property of $\psi$ gives the target $f_{\max}$. However, the functional value for $f$ is
$$
I_\phi(f) = \int_{\mathbb{R}^d} \phi(f(x))dx = \lambda^{-n}F_n(b) = \frac{F_n(b)}{G_n(b)} > A,
$$
where the last inequality is the violation of the condition (i). 

If the condition (ii) is violated, consider the density
$
f(x) = f_{\max}\cdot \mathbbm{1}(x\in [0, f_{\max}^{-1/n}]^n ). 
$
This is clearly a $\psi$-convex density with sup-norm $f_{\max}$. However, the functional value of $f$ is
$$
I_\phi(f) = \frac{\phi(f_{\max})}{f_{\max}} = \frac{\phi(\psi(b))}{\psi(b)} = \frac{F_0(b)}{G_0(b)} > A, 
$$
where the last inequality is the violation of condition (ii). 

\section{Application in Distributed Simulation}\label{sec.distributed}
In this section, we apply Theorem \ref{thm.upper} to the distributed simulation problem and show that finite bits of shared randomness are sufficient to exactly simulate an $n$-dimensional continuous density in a distributed fashion, provided that the joint density is $\beta$-concave with $\beta>n$ and has a finite generalized mutual information. This generalizes the result of \cite{li2017distributed} with log-concave densities. 

\subsection{Background}
In distributed simulation, let $\mathcal{X}$ be a possibly continuous alphabet, and $P$ be a given joint distribution on $\mathcal{X}^n$. A group of $n$ users aim to generate a random vector $(X_1,\cdots,X_n)\in \mathcal{X}^n$ distributed as $P$, where for each $i\in [n]$, user $i$ outputs $X_i$ based on her unlimited private randomness and some common randomness $W$ shared among all users. The goal is to characterize the minimum average description length of $W$, or effectively the minimum entropy $H(W)$, required for this task.

The solution to the minimum entropy is known to be the exact common information \cite{kumar2014exact} defined as
$$
G(P) \triangleq \min_{W: X_1 \perp X_2 \perp \cdots \perp X_n | W} H(W), 
$$
which is a generalization of Wyner’s common information \cite{wyner1975common} for exact simulation. Since $G(P)$ is the minimum of a concave function over a non-convex domain, it is computationally hard to evaluate in general. However, \cite{li2017distributed} shows that if $P$ admits a log-concave probability density on $\mathbb{R}^n$, the following upper and lower bounds of $G(P)$ are available: 
\begin{align}\label{eq.G_P}
I_D(P) \le G(P) \le I_D(P) + n^2 + 9n\log n, 
\end{align}
where $I_D(P)$ is the \emph{dual total correlation} of the joint density: 
\begin{align*}
I_D(P) \triangleq h(X) - \sum_{i=1}^n h(X_i | (X_j)_{j\neq i}). 
\end{align*}
The dual total correlation generalizes the mutual information in the sense that $I_D(P)$ reduces to $I(X_1;X_2)$ when $n=2$. As a result, the upper bound in \eqref{eq.G_P} shows that a finite amount of common randomness suffices to simulate any multivariate Gaussian distribution with a non-singular covariance. We also refer to \cite{yu2018exact} for an improved bound for Gaussian distributions with large blocks.

In the above result, besides the benign quasi-concave property of the density, we remark that the log-concave property is only used for the constrained value of its differential entropy, as shown in \eqref{eq.bobkov}. Since Theorem \ref{thm.upper} shows that the constrained value phenomenon for density functionals holds under general convexity conditions, it is expected that for weaker notions of convexity such as $\beta$-concavity, similar result to \eqref{eq.G_P} still holds. 

\subsection{Main Results}
The main result of this section is as follows. 
\begin{theorem}\label{thm.simulation}
Let $P$ be supported on $\mathbb{R}^n$ and admit a $\beta$-concave density with integer $\beta\ge 2n$. Then 
$$
I_D(P) \le G(P) \le I_D(P) + 2n^2 + 20n\log n.
$$
\end{theorem}
\begin{remark}
A careful inspection of the proof shows that $\beta\ge n+\varepsilon$ for any $\varepsilon>0$ suffices for Theorem \ref{thm.simulation}, with the additive gap explicitly depending on $(n,\varepsilon)$. 
\end{remark}

The proof of Theorem \ref{thm.simulation} is based on the following upper bound on $G(P)$ for quasi-concave densities, which is a direct consequence of \cite[Theorem 2]{li2017distributed}. 
\begin{lemma}\label{lemma.GP}
For a random vector $X\in \mathbb{R}^n \sim P$ with a quasi-concave density $f$, we have
\begin{align*}
G(P) & \le \tilde{h}_{1/(n+1)}(X) + \sum_{i=1}^n \log \int_{\mathbb{R}^{n-1}}\sup_{\tilde{x}_i}f(x_{\backslash i}, \tilde{x}_i) dx_{\backslash i} \\
& \qquad + (n+1)(1+2\log 2+\log(n+1)),
\end{align*}
where $\tilde{h}_\gamma(X)$ denotes the differential entropy of the truncated density $\gamma^{-1}\min\{f(x),t\}$, and $t>0$ is the solution to the equation $\int_{\mathbb{R}^n} \min\{f(x),t\}dx=\gamma$. 
\end{lemma}

We upper bound the first two terms of Lemma \ref{lemma.GP} separately. For the first term $\tilde{h}_{1/(n+1)}(X)$, since the truncated density is still $\beta$-concave, inequality \eqref{eq.entropy_beta} in Example \ref{example.entropy} gives
$$
\tilde{h}_{1/(n+1)}(X) \le \log\left(\frac{1}{(n+1)t} \right) + \sum_{i=1}^n \frac{\beta}{\beta-i} \le \log\left(\frac{1}{t}\right) + 2n,
$$
where the second inequality is due to $\beta\ge 2n$. On the other hand, $h(X)\ge \log(1/f_{\max})$, hence
\begin{align}\label{eq.first_term}
\tilde{h}_{1/(n+1)}(X) - h(X) \le \log\left(\frac{f_{\max}}{t}\right) + 2n. 
\end{align}
The next lemma presents an upper bound of $f_{\max}/t$ based on the results derived in Example \ref{example.truncated}. 

\begin{lemma}\label{lemma.truncation}
If $f$ is $\beta$-concave with integer $\beta\ge 2n$, we have $\log(f_{\max}/t)\le 6n$. 
\end{lemma}
\begin{proof}
Let $\phi_t(x) = \min\{x,t\}$ as in Example \ref{example.truncated}, then the definition of $t$ gives $I_{\phi_t}(f) = 1/(n+1)$. Hence, if $\log(f_{\max}/t) > 6n$, a combination of \eqref{eq.truncation_beta}, Lemma \ref{lemma.bino} and Hoeffding's inequality gives
\begin{align*}
\frac{1}{n+1} = I_{\phi_t}(f) &\le \mathbb{P}\left(\mathsf{B}\left(\beta, 1-\left(\frac{t}{f_{\max}}\right)^{\frac{1}{\beta}}\right) \le n\right)\\
&\le \mathbb{P}\left(\mathsf{B}\left(2n, 1-\left(\frac{t}{f_{\max}}\right)^{\frac{1}{2n}}\right) \le n\right) \\
&\le \mathbb{P}(\mathsf{B}(2n,19/20) \le n) \\
&\le \exp(-4n(19/20 - 1/2)^2) \le \exp(-0.8n),
\end{align*}
which is impossible for any $n\ge 1$. 
\end{proof}

To upper bound the second term of Lemma \ref{lemma.GP}, we need the following intermediate result in the proof of \cite[Lemma 3]{li2017distributed}. 
\begin{lemma}\label{lemma.sup}
For any quasi-concave density $f$ on $\mathbb{R}^n$, 
\begin{align*}
\left(\sup_{x^{n-1}} \int_{\mathbb{R}}f(x)dx_n\right)\left(\int_{\mathbb{R}^{n-1}}\sup_{x_n} f(x)dx^{n-1} \right) \le n\sup_x f(x). 
\end{align*}
\end{lemma}

Based on Lemma \ref{lemma.sup} and inequality \eqref{eq.entropy_beta}, we have
\begin{align}
&\log \int_{\mathbb{R}^{n-1}}\sup_{x_n} f(x)dx^{n-1}  \nonumber \\
&\quad \le \log n + \log\left(\frac{1}{\sup_{x^{n-1}} \int_{\mathbb{R}}f(x)dx_n}\right) - \log\left(\frac{1}{\sup_x f(x)}\right) \nonumber \\
&\quad \le \log n + h(X^{n-1}) - h(X) + \sum_{i=1}^n \frac{\beta}{\beta-i} \nonumber \\
&\quad \le - h(X_n|X^{n-1}) + \log n + 2n. \label{eq.second_term}
\end{align}
Now the combination of Lemma \ref{lemma.GP}, Lemma \ref{lemma.truncation}, \eqref{eq.first_term} and \eqref{eq.second_term} completes the proof of Theorem \ref{thm.simulation}. 

\section{Acknowledgement}
Yanjun Han would like to thank Prof. Abbas El Gamal from Stanford University for pointing out the distributed simulation problem and many helpful suggestions.

\bibliographystyle{IEEEtran}
\bibliography{di}

\newcommand{\noopsort}[1]{}
\begin{thebibliography}{10}
\providecommand{\url}[1]{#1}
\csname url@samestyle\endcsname
\providecommand{\newblock}{\relax}
\providecommand{\bibinfo}[2]{#2}
\providecommand{\BIBentrySTDinterwordspacing}{\spaceskip=0pt\relax}
\providecommand{\BIBentryALTinterwordstretchfactor}{4}
\providecommand{\BIBentryALTinterwordspacing}{\spaceskip=\fontdimen2\font plus
\BIBentryALTinterwordstretchfactor\fontdimen3\font minus
  \fontdimen4\font\relax}
\providecommand{\BIBforeignlanguage}[2]{{%
\expandafter\ifx\csname l@#1\endcsname\relax
\typeout{** WARNING: IEEEtran.bst: No hyphenation pattern has been}%
\typeout{** loaded for the language `#1'. Using the pattern for}%
\typeout{** the default language instead.}%
\else
\language=\csname l@#1\endcsname
\fi
#2}}
\providecommand{\BIBdecl}{\relax}
\BIBdecl

\bibitem{bobkov2011entropy}
S.~Bobkov and M.~Madiman, ``The entropy per coordinate of a random vector is
  highly constrained under convexity conditions,'' \emph{IEEE Transactions on
  Information Theory}, vol.~57, no.~8, pp. 4940--4954, 2011.

\bibitem{li2017distributed}
C.~T. Li and A.~El~Gamal, ``Distributed simulation of continuous random
  variables,'' \emph{IEEE Transactions on Information Theory}, vol.~63, no.~10,
  pp. 6329--6343, 2017.

\bibitem{prekopa1973logarithmic}
A.~Pr{\'e}kopa, ``On logarithmic concave measures and functions,'' \emph{Acta
  Scientiarum Mathematicarum}, vol.~34, pp. 335--343, 1973.

\bibitem{bakry1985diffusions}
D.~Bakry and M.~{\'E}mery, ``Diffusions hypercontractives,'' in
  \emph{S{\'e}minaire de Probabilit{\'e}s XIX 1983/84}.\hskip 1em plus 0.5em
  minus 0.4em\relax Springer, 1985, pp. 177--206.

\bibitem{bobkov1999isoperimetric}
S.~G. Bobkov, ``Isoperimetric and analytic inequalities for log-concave
  probability measures,'' \emph{The Annals of Probability}, vol.~27, no.~4, pp.
  1903--1921, 1999.

\bibitem{brascamp2002extensions}
H.~J. Brascamp and E.~H. Lieb, ``On extensions of the brunn-minkowski and
  pr{\'e}kopa-leindler theorems, including inequalities for log concave
  functions, and with an application to the diffusion equation,'' in
  \emph{Inequalities}.\hskip 1em plus 0.5em minus 0.4em\relax Springer, 2002,
  pp. 441--464.

\bibitem{bobkov2012reverse}
S.~Bobkov and M.~Madiman, ``Reverse brunn--minkowski and reverse entropy power
  inequalities for convex measures,'' \emph{Journal of Functional Analysis},
  vol. 262, no.~7, pp. 3309--3339, 2012.

\bibitem{xu2016reverse}
P.~Xu, J.~Melbourne, and M.~M. Madiman, ``Reverse entropy power inequalities
  for s-concave densities.'' in \emph{ISIT}, 2016, pp. 2284--2288.

\bibitem{li2018renyi}
J.~Li, ``R{\'e}nyi entropy power inequality and a reverse,'' \emph{Studia
  Mathematica}, vol. 242, pp. 303--319, 2018.

\bibitem{borell1973complements}
C.~Borell, ``Complements of lyapunov's inequality,'' \emph{Mathematische
  Annalen}, vol. 205, no.~4, pp. 323--331, 1973.

\bibitem{fradelizi2020concentration}
M.~Fradelizi, J.~Li, M.~Madiman \emph{et~al.}, ``Concentration of information
  content for convex measures,'' \emph{Electronic Journal of Probability},
  vol.~25, 2020.

\bibitem{morris1963note}
K.~Morris, ``A note on direct and inverse binomial sampling,''
  \emph{Biometrika}, vol.~50, no. 3-4, pp. 544--545, 1963.

\bibitem{kumar2014exact}
G.~R. Kumar, C.~T. Li, and A.~El~Gamal, ``Exact common information,'' in
  \emph{2014 IEEE International Symposium on Information Theory}.\hskip 1em
  plus 0.5em minus 0.4em\relax IEEE, 2014, pp. 161--165.

\bibitem{wyner1975common}
A.~Wyner, ``The common information of two dependent random variables,''
  \emph{IEEE Transactions on Information Theory}, vol.~21, no.~2, pp. 163--179,
  1975.

\bibitem{yu2018exact}
L.~Yu and V.~Y. Tan, ``On exact and $\infty$-{R}\'{e}nyi common informations,''
  \emph{arXiv preprint arXiv:1810.00295}, 2018.

\end{thebibliography}

\end{document}